\documentclass[12pt]{article}
\usepackage{a4}
\usepackage{array}
\usepackage{amssymb}
\usepackage{amsfonts}
\usepackage{amsmath}

\newtheorem{theorem}{Theorem}[section]
\newtheorem{lemma}[theorem]{Lemma}

\newtheorem{corollary}[theorem]{Corollary}

\newcommand{\kom}[1]{}

\newcommand{\proofend}{\rule{1em}{0in} \hspace*{\fill}$\square$\vspace{1ex}\par}
\newenvironment{proof}{\par\noindent{\em Proof\/}.~}{\proofend}

\newcommand{\set}[2]{\{\,#1\mid#2\,\}}

\title{The Quantum Query Complexity \\ of Algebraic Properties
}
\author{Sebastian D\"orn \\
\normalsize Institut f\"{u}r Theoretische Informatik\\
\normalsize Universit\"{a}t Ulm\\
\normalsize 89069 Ulm, Germany
\and
Thomas Thierauf\\
\normalsize Fak.\ Elektronik und Informatik\\
\normalsize  HTW Aalen\\
\normalsize 73430 Aalen, Germany
\\
\and
\normalsize ~~~~~~~~\{sebastian.doern,thomas.thierauf\}@uni-ulm.de~~~~~~~~
}
\date{}

\begin{document}
\maketitle

\begin{abstract}
We present quantum query complexity bounds for testing algebraic properties.
For a set $S$ and a binary operation on~$S$,
we consider the decision problem whether~$S$ is a semigroup or has an identity element.
If~$S$ is a monoid, we want to decide whether~$S$ is a group.

We present quantum algorithms for these problems 
that improve the best known classical complexity bounds.
In particular,
we give the first application of the new quantum random walk technique by  
Magniez, Nayak, Roland, and Santha~\cite{MNRS07}
that improves the previous bounds by Ambainis~\cite{Amb04} and Szegedy~\cite{Sze04}.
We also present several lower bounds for testing algebraic properties.
\end{abstract}


\section{Introduction}\label{s:intro}

Quantum algorithms have the potential to demonstrate 
that for some problems quantum computation is more efficient than classical computation. 
A goal of quantum computing is to determine whether quantum computers are faster than classical computers. 

In search problems, the access to the input is done via an oracle. 
This motivates the definition of the query complexity, which measures the number of accesses to the oracle. 
Here we study the quantum query complexity, which is the number of quantum queries to the oracle.
For some problems the quantum query complexity can be exponentially smaller than the classical one;
an example is the Simon algorithm~\cite{Sim94}. 

Quantum query algorithms have been presented for several problems, 
see~\cite{BDHHMSW01,Amb04,DHHM04,MN05,MSS05,BS06,Doe07a,Doe07b}.
These algorithms use search techniques like Grover search~\cite{Gro96}, amplitude amplification~\cite{BHMT00} and quantum random walk~\cite{Amb04,Sze04}.

In this paper we study the quantum query complexity for testing algebraic properties. 
Our input is a  multiplication table for a set~$S$ of size $n \times n$.
In Section~\ref{s:semigroup} we consider the {\em  semigroup problem\/},
that is, whether the operation on~$S$ is associative.
Rajagopalan and Schulman~\cite{RS00} developed a randomized algorithm 
for this problem that runs in time $O(n^2)$.
As an additional parameter, 
we consider the binary operation $\circ: S \times S \rightarrow M$, where $M \subseteq S$.
We construct a quantum algorithm for this problem whose query complexity is $O(n^{5/4})$, 
if the size of $M$ is constant. 
Our algorithm is the first application of the new quantum random walk search scheme 
by Magniez, Nayak, Roland, and Santha~\cite{MNRS07}.
With the quantum random walk of Ambainis~\cite{Amb04} and Szegedy~\cite{Sze04},
the query complexity of our algorithm would not improve the obvious Grover search algorithm for this problem.
We show  a quantum query lower bound for the semigroup  problem of $\Omega(n)$ in Section~\ref{s:lowerbound}. 

In Section~\ref{s:group} we consider the {\em group problem},
that is, whether the monoid~$M$ given by its multiplication table is a group.
We present a  randomized algorithm that solves 
the problem with $O(n^{\frac{3}{2}})$ classical queries to the multiplication table.
This improves the naive $O(n^2)$ algorithm that searches for an inverse in the 
multiplication table for every element.
Then we show that on a  quantum computer  
the query complexity can be improved to $\widetilde{O}(n^{\frac{11}{14}})$,
where the $\widetilde{O}$-notation hides a logarithmic factor.

In Section~\ref{s:lowerbound},
we show linear lower bounds for the semigroup problem and the identity problem.
In the latter problem we have given a multiplication table of a set~$S$
and have to decide
whether~$S$ has an identity element.
As an upper bound,
the identity problem can be solved with linearly many quantum queries,
which matches the lower bound.
Finally we show linear lower bounds for the quasigroup and the loop problem,
where one has to decide whether a multiplication table is a quasi group or a loop, respectively.


\section{Preliminaries}\label{s:prel}

\subsection{Quantum Query Model}
In the query model, 
the input $x_1,\ldots , x_N$ is contained in a black box or oracle 
and can be accessed by queries to the black box. 
As a query we give~$i$ as input to the black box and the black box outputs~$x_i$. 
The goal is to compute a Boolean function $f:\{0, 1\}^N \rightarrow \{0, 1\}$ 
on the input bits $x=(x_1, \ldots, x_N)$ minimizing the number of queries. 
The classical version of this model is known as decision tree.

The quantum query model was explicitly introduced by Beals {\it et al.\/}~\cite{BBCMW01}. 
In this model we pay for accessing the oracle, 
but unlike the classical case, 
we use the power of quantum parallelism to make queries in superposition.
The state of the computation is represented by $\left|i, b, z\right\rangle$, 
where~$i$ is the query register,
$b$ is the answer register, 
and $z$ is the working register. 

A 
quantum computation with $k$ queries is a sequence of unitary transformations
\[U_0 \rightarrow O_x \rightarrow  U_1 \rightarrow  O_x \rightarrow  \ldots 
\rightarrow  U_{k-1} \rightarrow  O_x \rightarrow  U_k,\]
where each $U_j$ is a unitary transformation 
that does not depend on the input~$x$, 
and $O_x$ are query (oracle) transformations. 
The oracle transformation $O_x$ can be defined as
$O_x :\left|i, b, z\right\rangle \rightarrow  \left|i, b \oplus x_i, z\right\rangle$.

The computation consists of the following three steps: 
\begin{enumerate}
\item
Go into the initial state $\left|0 \right\rangle$.
\item
Apply the transformation $U_T O_x \cdots O_x U_0$.
\item
Measure the final state.
\end{enumerate}
The result of the computation is the rightmost bit 
of the state obtained by the measurement.

The quantum computation determines~$f$ with bounded error, 
if for every $x$, 
the probability that the result of the computation
equals $f(x_1, \ldots , x_N)$ is at least $1-\epsilon$, for some fixed $\epsilon < 1/2$.
In the query model of computation each query adds one to the query complexity of an algorithm, 
but all other computations are free. 


\subsection{Tools for Quantum Algorithms}

For the basic notation on quantum computing, 
we refer the reader to the textbook by Nielsen and Chuang [NC03].
Here, we give three tools for the construction of our quantum  algorithms.


\paragraph{Quantum Search.}

A search problem is a subset $P \subseteq \{1, \ldots, N\}$ of the search space~$\{1, \ldots, N\}$.
With~$P$ we associate its characteristic function \linebreak $f_P: \{1, \ldots, N\} \rightarrow \{0,1\}$ with
\[
f_P(x)=
\begin{cases}
1, & \text{if } x \in P,\\
0, & \text{otherwise}.
\end{cases}
\]
Any $x \in P$ is called a solution to the search problem.
Let $k = |P|$ be the number of solutions of~$P$.

\begin{theorem}{\rm \cite{Gro96,BBHT98}}\label{Grover}
For $k > 0$,
the expected quantum query complexity for finding one solution of $P$ is $O(\sqrt{N/k})$, 
and for finding all solutions, it is $O(\sqrt{kN})$. 
Futhermore, 
whether $k>0$ can be decided  in $O(\sqrt{N})$ quantum queries to $f_P$. 
\end{theorem}


\paragraph{Amplitude Amplification.}

Let $\mathcal A$ be an algorithm for a problem with small success probability at least $\epsilon$.
Classically, we need $\Theta(1/\epsilon)$ repetitions of $\mathcal A$ to increase 
its success probability from $\epsilon$ to a constant, for example 2/3.  
The corresponding technique in the quantum case
is called amplitude amplification.

\begin{theorem}{\rm \cite{BHMT00}}\index{Amplitude amplification}
Let $\mathcal A$ be a quantum algorithm with one-sided error and success probability at
least $\epsilon$. 
Then there is a quantum algorithm $\mathcal B$ that solves $\mathcal A$ with success
probability 2/3 by $O(\frac{1}{\sqrt{\epsilon}})$ invocations of $\mathcal A$.
\end{theorem}


\paragraph{Quantum Walk.}

Quantum walks are the quantum counterpart of Markov chains and random walks. 
The quantum walk search provide a promising source for new quantum algorithms, like
quantum walk search algorithm~\cite{KSW03}, 
element distinctness algorithm~\cite{Amb04}, 
triangle finding~\cite{MSS05}, 
testing group commutativity~\cite{MN05}, 
and matrix verification~\cite{BS06}.  

Let~$P = (p_{xy})$ be the transition matrix of an ergodic symmetric Markov chain on the state space~$X$.
Let $M \subseteq X$ be a set of marked states.  
Assume that the search algorithms use a data structure~$D$ 
that associates some data~$D(x)$ with every state $x \in X$.  
From $D(x)$, we would like to determine if $x \in M$.  
When operating on~$D$, we consider the following three types of cost:
\begin{description}     
\item [Setup cost] 
$s$: The worst case cost to compute $D(x)$, for $x\in X$. 
\item [Update cost] 
$u$: The worst case cost for transition from $x$ to $y$, and update $D(x)$ to $D(y)$. 
\item [Checking cost] 
$c$: The worst case cost for checking if $x \in M$ by using $D(x)$.
\end{description}

Magniez {\it et al.\/}~\cite{MNRS07} developed a new scheme for quantum search, 
based on any ergodic Markov chain. 
Their work generalizes previous results by Ambainis~\cite{Amb04} and Szegedy~\cite{Sze04}. 
They extend the class of possible Markov chains 
and improve the query complexity as follows.

\begin{theorem}{\rm \cite{MNRS07}} \label{MNRS07}\index{Quantum Walk!Schema}
Let $\delta > 0$ be the eigenvalue gap of a ergodic Markov chain $P$ 
and let $\frac{|M|}{|X|} \geq \epsilon$.
Then there is a quantum algorithm that determines if $M$ is empty or finds an element of $M$ with cost
\[
s + \frac{1}{\sqrt{\epsilon}} \left( \frac{1}{\sqrt{\delta}} u + c\right).
\]
\end{theorem}

In the most practical application (see~\cite{Amb04,MSS05}) the quantum walk takes place on the Johnson graph $J(n,r)$, 
which is defined as follows: 
the vertices are subsets of $\{1, \ldots, n\}$ of size $r$ 
and two vertices are connected iff they differ in exactly one number. 
It is well known, that the spectral gap $\delta$ of $J(n,r)$ 
is $1/r$. 


\subsection{Tool for Quantum Query Lower Bounds}

In this paper, we use the following special case of a method
by Ambainis~\cite{Amb02} to prove lower bounds for the quantum query complexity.

\begin{theorem}{\rm \cite{Amb02}}\label{Amb02}
Let $A \subset \{0, 1\}^n, B\subset \{0, 1\}^n$ and $f:\{0,1\}^n \rightarrow \{0,1\}$ 
such that $f(x)=1$ for all $x\in A$, and $f(y)=0$ for all $y \in B$. 
Let~$m$ and~$m'$ be numbers such that
\begin{enumerate}
\item 
for every $(x_1,\ldots,x_n)\in A$ there are at least $m$ values $i \in \{1,\ldots, n\}$ 
such that $(x_1, \ldots , x_{i-1}, 1 - x_i, x_{i+1}, \ldots , x_n)\in B$,
\item 
for every $(x_1,\ldots,x_n)\in B$ there are at least $m'$ values $i \in \{1,\ldots, n\}$ 
such that $(x_1, \ldots, x_{i-1}, 1 - x_i, x_{i+1}, \ldots , x_n)\in A$.
\end{enumerate}
Then every bounded-error quantum algorithm that computes $f$ has quantum query complexity
$\Omega(\sqrt{m \cdot m'})$.
\end{theorem}

\section{The Semigroup Problem}\label{s:semigroup}

In the semigroup problem we have given two sets~$S$ and $M \subseteq S$ 
and a binary operation $\circ: S \times S \rightarrow M$
represented by  a table.
We denote with $n$ the size of the set $S$.
One has to decide whether~$S$ is a semigroup,
that is, whether the operation on~$S$ is associative.

The complexity of this problem was first considered by Rajagopalan and Schulman~\cite{RS00}, 
who gave a randomized algorithm with time complexity of $O(n^2 \log \frac{1}{\delta})$, 
where $\delta$ is the error probability. 
They also showed a lower bound of $\Omega(n^2)$.
The previously best known algorithm was the naive $\Omega(n^3)$-algorithm that checks all triples. 

In the quantum setting, one can do a Grover search over all triples $(a,b,c) \in S^3$
and check whether the triple is associative.
The quantum query complexity of the search is $O(n^{3/2})$.
We construct a quantum algorithm for the semigroup problem that has query complexity $O(n^{5/4})$, 
if the size of $M$ is constant. 
In Section~\ref{s:lowerbound} we give a quantum query lower bound of $\Omega(n)$ for this problem.

Our algorithm is the first application of the recent quantum random walk search scheme 
by Magniez {\it et al.\/} \cite{MNRS07}. 
The quantum random walk of Ambainis~\cite{Amb04}
and Szegedy~\cite{Sze04} doesn't suffice to
get an improvement of the Grover search  mentioned above.

\begin{theorem}
Let $k=n^{\alpha}$ be the size of $M$ with $0 \leq \alpha \leq 1$. 
The quantum query complexity of the semigroup problem is 
\[
\begin{cases}
O(n^{\frac{5}{4} + \frac{\alpha}{2}}), & \text{for } 0 < \alpha \leq \frac 1 6,\\  
O(n^{\frac{6}{5} + \frac{4}{5}\alpha}), & \text{for } \frac{1}{6} < \alpha \leq \frac{3}{8},\\  
O(n^{\frac{3}{2}}), & \text{for } \frac{3}{8} < \alpha \leq 1. 
\end{cases}
\]
\end{theorem}

\begin{proof}
We use the quantum walk search scheme of Theorem~\ref{MNRS07}. 
To do so,
we construct a Markov chain and a database for checking if a vertex of the chain is marked. 

Let $A$ and $B$ two subsets of~$S$ of size~$r$ that are disjoint from~$M$.
We will determine~$r$ later.
The database is the set 
\[
                D(A,B)=\set{(a,b, a \circ b)}{a \in A \cup M \text{ and } b \in B \cup M}.
\]
Our quantum walk is done on the  categorical graph product of two 
Johnson graphs $G_J=J(n-k,r) \times J(n-k,r)$. 
The marked vertices of $G_J$ correspond to pairs 
$(A,B)$ with $(A \circ B) \circ S \neq A \circ (B \circ S)$.  
In every step of the walk, we exchange one row and one column of~$A$ and~$B$.
 
Now we compute the quantum query costs for the setup, update and checking. 
The setup cost for the database $D(A,B)$ is $(r+k)^2$ and the update cost is $r+k$. 
To check whether a pair $(A,B)$ is marked, 
we search for a pair 
$(b,c)\in B \times S$ with $(A\circ b) \circ c \neq A \circ (b \circ c)$. 
The quantum query cost to check this inequality is $O(k)$, by using our database. 
Therefore, by applying Grover search, the checking cost is $O(k \sqrt{n r})$.
The spectral gap of the walk on $G_J$ is 
$\delta = O(1/r)$ for $1 \le r \le \frac{n}{2}$, see~\cite{BS06}.
If there is a triple $(a,b,c)$ with $(a \circ b) \circ c \neq a \circ (b \circ c)$, 
then there are at least  $\binom{n-k-1}{r-k-1}^2$ marked sets $(A, B)$.
Therefore we have
\[
\epsilon 
\geq \frac{|M|}{|X|} 
\geq \left(\frac{\binom{n-k-1}{r-k-1}}{\binom{n-k}{r-k}}\right)^2 
\geq \left(\frac{r-k}{n-k}\right)^2.
\]
Let $r=n^{\beta}$, for $0 < \beta < 1$.
Assuming $r>2k$ we have
\[
\frac{1}{\sqrt{\epsilon}} \leq \frac{n-k}{r-k} \leq \frac{n}{r/2} = \frac{2n}{r}.
\]
Then the quantum query complexity of the semigroup problem is
\[
O\left(r^2 + \frac{n}{r} \left(\sqrt{r} \cdot r + \sqrt{n r} \cdot k \right)\right)
= O\left( n^{2\beta} + n^{1+\frac{\beta}{2}} + n^{\frac{3}{2} + \alpha - \frac{\beta}{2}}\right).
\]
Now we choose~$\beta$ depending on~$\alpha$ such that this expression is minimal.
A straight forward calculation gives the bounds claimed in the theorem.
\end{proof}

For the special case that $\alpha = 0$,
i.e., only a constant number of elements occurs in the multiplication table,
we get
\begin{corollary}
The quantum query complexity of the semigroup problem is  $O(n^{\frac{5}{4}})$, if $M$ has constant size.
\end{corollary}

Note that the time complexity of our algorithm is $O(n^{1.5} \log n)$.

\section{Group Problems}\label{s:group}

In this section we consider the problem whether a given finite monoid $M$ is in fact a group.
That is, we have to check whether every element of~$M$ has an inverse.
The monoid~$M$ has $n$ elements and is given by its multiplication table
and the identity element~$e$.

To the best of our knowledge, 
the group problem has not been studied before.
The naive approach for the problem checks for every element~$a \in M$,
whether $e$ occurs in $a$'s row in the multiplication table.
The query complexity is $O(n^2)$.
We develop a (classical) randomized algorithm that solves the problem with $O(n^{\frac{3}{2}})$ queries
to the multiplication table.
Then we show that on a  quantum computer  
the query complexity can be improved to $\widetilde{O}(n^{\frac{11}{14}})$.

\begin{theorem}
Whether a given monoid with $n$ elements is a group can be decided with query complexity
\begin{enumerate}
\item 
$O(n^{\frac{3}{2}})$ by a randomized algorithm with probability $\geq 1/2$,
\item
$O(n^{\frac{11}{14}} \log n)$ by a quantum query algorithm.
\end{enumerate}
\end{theorem}

\begin{proof}
Let $a \in M$.
We consider the sequence of powers $a, a^2, a^3, \dots$.
Since $M$ is finite, there will be a repetition at some point.
We define the {\em order of\/}~$a$ as the smallest power~$t$,
such that $a^t = a^s$, for some $s < t$.
Clearly, if $a$ has an inverse, $s$ must be zero.

\begin{lemma}\label{le:order}
Let $a \in M$ of order~$t$.
Then $a$ has an inverse iff $a^t = e$.
\end{lemma}

Hence the powers of~$a$ will tell us at some point whether $a$ has an inverse.
On the other hand,
if $a$ has no inverse,
the powers of~$a$ provide more elements with no inverse as well.

\begin{lemma}\label{le:inverse}
Let $a \in M$.
If $a$ has no inverse, then $a^k$ has no inverse, for all $k \geq 1$.
\end{lemma}

Our algorithm has two phases.
In phase~1,
it computes the powers of every element up to certain number~$r$.
That is, we consider the sequences $S_r(a) = (a, a^2,  \dots, a^r )$, for all $a \in M$.
If $e \in S_r(a)$ then $a$ has an inverse by Lemma~\ref{le:order}.
Otherwise,
if we find a repetition in the sequence $S_r(a)$,
then, again by Lemma~\ref{le:order}, $a$ has no inverse  and we are done.

If we are not already done by phase~1,
i.e.\ there are some sequences $S_r(a)$ left such that 
$e \not\in S_r(a)$ and $S_r(a)$ has pairwise different elements,
then the algorithm proceeds to phase~2.
It selects some $a \in M$ uniformly at random
and checks whether~$a$ has an inverse by searching for~$e$ in the row of~$a$ in the multiplication table.
This step is repeated $n/r$ times.

The query complexity $t(n)$ of the algorithm is bounded by $nr$ in phase~1 and by $n^2/r$ in phase~2.
That is $t(n) \leq  nr + n^2/r$, which is minimized for $r = n^{\frac{1}{2}}$.
Hence we have $t(n) \leq 2n^{\frac{3}{2}}$.

For the correctness  observe that the algorithm accepts with probability~1 if $M$ is a group.
Now assume that $M$ is not a group.
Assume further that the algorithm does not already detect this  in phase~1.
Let $a$ be some element without an inverse.
By Lemma~\ref{le:order},
the sequence $S_r(a)$ has $r$ pairwise different elements
which don't have inverses too
by Lemma~\ref{le:inverse}.
Therefore in phase~2, 
the algorithm  picks an element without an inverse with probability at least $r/n$.
By standard arguments,
the probability that at least one out of $n/r$ many randomly chosen elements has no inverse is constant.

For the quantum query complexity we use Grover search and amplitude amplification.
In phase~1,
we search for an $a \in M$, such that 
the sequence $S_r(a)$ has~$r$ pairwise different entries different from~$e$.
This property can be checked by 
first searching $S_r(a)$ for an occurance of~$e$ by a Grover search with $\sqrt{r} \log r$ queries.
Then, if~$e$ doesn't occur in $S_r(a)$,
we check whether there is an element in $S_r(a)$ that occurs more than once.
This is the element distinctness problem and can be solved with $r^{2/3} \log r$ queries, see~\cite{Amb04}.
Therefore the quantum query complexity of phase~1  is bounded by $\sqrt{n} \cdot r^{2/3} \log r$.

In phase~2 we search for an $a \in M$ such that $a$ has no inverse.
In phase~2 we actually search the row of~$a$ in the multiplication table.
Hence this takes $\sqrt{n}$ queries.
Since at least $r$ of the $a$'s don't have an inverse,
by amplitude amplification we get $\sqrt{n} \sqrt{n/r} = n/\sqrt{r}$ queries in phase~2.

In summary,
the quantum query complexity is 
$\sqrt{n} \cdot r^{2/3} \log r + n/\sqrt{r}$, 
which is minimized for $r = n^{\frac{3}{7}}$. 
Hence we have a $O(n^{\frac{11}{14}} \log n)$ quantum query algorithm.
\end{proof}


\section{Lower Bounds}\label{s:lowerbound}

\begin{theorem}
The semigroup problem requires $\Omega(n)$ quantum queries.
\end{theorem}

\begin{proof}
Let~$S$ be a set of size $n$ and $\circ: S \times S \rightarrow \{0,1\}$ a binary operation represented by a table.
We apply Theorem \ref{Amb02}.
The set $A$ consists of all $n\times n$ matrices, where the entry of position $(1,1),
(1,c), (c,1)$ and $(c,c)$ is 1, 
for $c\in S - \{0,1\}$, and zero otherwise. 
It is easy to see, 
that the multiplication tables of~$A$ are associative, 
since $(x \circ y) \circ z = x \circ (y \circ z)=1$ for all $x,y,z \in \{1,c\}$ and zero otherwise.

The set $B$ consists of all $n \times n$ matrices, where the entry of position $(1,1),
(1,c), (c,1)$, $(c,c)$ and $(a,b)$ is~1, 
for fixed $a, b, c\in S- \{0,1\}$ with $a, b\neq c$, and zero otherwise. 
Then $(a \circ b) \circ c = 1$ and $a \circ (b \circ c)=0$. 
Therefore the multiplication tables of~$B$ are not associative.

From each $T \in A$, 
we can obtain $T'\in B$ by replacing the  entry~0 of $T$ at $(a,b)$ by~1,
for any $a,b \notin \{0,1,c\}$.
Hence we have $m=\Omega(n^2)$.  
From each $T'\in B$, we
can obtain $T\in A$ by replacing the entry~1 of $T'$ at position $(a,b)$ by~0,
for  $a, b\notin \{0,1,c\}$.
Then we have $m' = 1$.
By Theorem~\ref{Amb02}, 
the quantum query complexity is $\Omega(\sqrt{m \cdot m'}) = \Omega(n)$.
\end{proof}

Next, 
we consider the {\em identity problem\/}:
given the multiplication table on a set~$S$,
decide whether there is an identity element.\footnote{Here we consider right identity, the case of left identity is analogous.}
We show that the identity problem requires linearly many quantum queries. 
We start by considering the {\em 1-column problem\/}:
given a 0-1-matrix of order~$n$, 
decide whether it contains a column that is all~1.

\begin{lemma}
The 1-column problem requires $\Omega(n)$ quantum queries.
\end{lemma}
\begin{proof} 
We use Theorem~\ref{Amb02}. 
The set~$A$ consists of all matrices, 
where in $n-1$ columns there is exactly one entry with value~0, 
and the other entries of the matrix are~1. 
The set~$B$ consists of all matrices, 
where in every column there is exactly one entry with value~0, 
and the other entries of the matrix are~1. 
From each matrix $T \in A$, 
we can obtain $T'\in B$ by changing one entry in the 1-column from~1 to~0.
Then we have $m=n$.  
From each matrix $T'\in B$, 
we can obtain $T\in A$ by changing one entry from~0 to~1.
Then we have  $m' = n$.
By Theorem \ref{Amb02}, the quantum query complexity is $\Omega(n)$.
\end{proof}

\begin{theorem}
The identity problem requires $\Omega(n)$ quantum queries.
\end{theorem}

\begin{proof} 
We reduce the 1-column problem to the identity problem.
Given a 0-1-matrix $M = \left( m_{i,j} \right)$ of order~$n$.
We define $S = \{0, 1, \dots, n\}$ and a multiplication table $T = \left( t_{i,j} \right)$ with $0 \leq i,j \leq n$ 
for~$S$ as follows:
\[t_{i,j} = 
\begin{cases}
0, & \text{if } m_{i,j} = 0,\\
i, & \text{if } m_{i,j} = 1,
\end{cases}
\]
and $t_{0, j} = t_{i,0}=0$.
Then $M$ has a 1-column iff $T$ has an identity element.
\end{proof}

Finding an identity element is simple. 
We choose an element $a\in S$ and then we test if 
$a$ is the identity element by using Grover search in $O(\sqrt{n})$ quantum queries. 
The success probability of this procedure is $\frac{1}{n}$. 
By using the amplitude amplification 
we get an $O(n)$ quantum query algorithm for finding an identity element (if there is one).
Since the upper and the lower bound match, we have determined the precise complexity of the identity problem.

\begin{corollary}
The identity problem has quantum query complexity $\Theta(n)$.
\end{corollary}

In the {\em quasigroup problem\/} we have given a set~$S$ and a binary operation on~$S$
represented by  a table.
One has to decide whether~$S$ is a quasigroup,
that is, whether all equations $a \circ x = b$ and $x \circ a = b$ 
have unique solutions.
In the {\em loop problem\/}, one has to decide whether~$S$ is a loop.
A loop is a quasigroup with an identity element~$e$
such that $a \circ e  = a = e \circ a$ for all $a \in S$.

In the multiplication table of a quasigroup, every row and column  is a permutation of the elements of~$S$.
In a loop, there must occur  the identity permutation in some row and some column.
We have already seen how to determine an identity element with $O(n)$ quantum queries.
A row or column is a permutation,
if no element appears twice.
Therefore one can use the element distinctness quantum algorithm by Ambainis~\cite{Amb04}
to search for a row or column with two equal elements.
The quantum query complexity of the search is $O(\sqrt{n}\cdot n^{\frac 2 3}) = O(n^{\frac 7 6})$.
We show in the following theorem
an $\Omega(n)$ lower bound for these problems.

\begin{theorem}
The quasigroup problem and the loop problem require $\Omega(n)$ quantum queries.
\end{theorem}

\begin{proof}
We reduce the {\em  identity matrix problem\/} to the loop problem.
Given a 0-1-matrix $M = \left( m_{i,j} \right)$ of order~$n$,
decide whether~$M$ is the identity matrix.
It is not hard to see that the  identity matrix problem requires $\Omega(n)$ quantum queries
(similar as for the 1-column problem).

We define $S = \{0,1, \dots, n-1\}$ and a multiplication table $T = \left( t_{i,j} \right)$ for~$S$.
For convenience, we take indices $0 \leq i,j \leq n-1$ for~$M$ and~$T$.
The entries of the second diagonal are
\[t_{i,n-1-i} = 
\begin{cases}
n-1, & \text{if } m_{i,i} = 1,\\
0, & \text{otherwise.}
\end{cases}
\]
For $j \not= n-1-i$ we define
\[t_{i,j} = 
\begin{cases}
(i + j )\bmod{n}, & \text{if } m_{i,n-1-j} = 0,\\
0, & \text{if } m_{i,n-1-j} = 1 \text{ and } (i+j ) \bmod{n} \not= 0,\\
1, & \text{otherwise.} 
\end{cases}
\]
If $M$ is the identity matrix,
then~$T$ is a circular permutation matrix
\[T = 
\begin{pmatrix}
0   & 1  &   \cdots  & n-2 & n-1\\
1 & 2   &  \cdots & n-1 &  0 \\
\vdots & \vdots &  & \vdots & \vdots  \\
n-1 & 0 &   \cdots & n-3 & n-2
\end{pmatrix}.\]
Hence $S$ is a loop with identity~$0$.

Suppose~$M$ is not the identity matrix.
If~$M$ has a~0 on the main diagonal,
say at position $(i,i)$,
then the value $n-1$ doesn't occur in row~$i$ in~$T$,
and hence, row~$i$ is not a permutation of~$S$.
If~$M$ has a~1 off the main diagonal, say at position  $(i, n-1 -j)$,
then there will be a 0 or 1 at position $(i,j)$ in $T$, which is different from $(i + j) \bmod{n}$.
Hence, either there will be two 0's or two 1's in row~$i$ in~$T$,
in which case row~$i$ is not a permutation of~$S$,
or a 0 or 1 changes to 1 or 0.
In the latter case, 
the $i$-th row of~$T$ can be a permutation only if correspondingly
the other 0 or 1 changes as well to 1 or 0, respectively.
But then this carries over to the other  rows and columns of~$T$.
That is,  there must be more 1's in~$M$ off the main diagonal,
so that all 0's and 1's in $T$ switch their place with respect to their position
when $M$ is the identity matrix.
However,
then there is no identity element in~$T$
and hence, $T$ is not a loop.

The reduction to the quasigroup problem can be done with similar arguments.
\end{proof}


\section*{Conclusion and Open Problems}
In this paper we present quantum query complexity bounds of algebraic properties.
We construct a quantum algorithm for the semigroup problem whose query complexity is 
$O(n^{5/4})$,
if the size of $M$ is constant. Then we consider the group problem, and presented a  
randomized algorithm that solves this problem with $O(n^{\frac{3}{2}})$ classical queries 
and $\widetilde{O}(n^{\frac{11}{14}})$ quantum queries to the multiplication table. 
Finally we show linear lower bounds for the semigroup, identity, quasigroup 
and loop problem.

Some questions remain open: Is there a quantum algorithm for the semigroup problem which is 
better then the Grover search bound of $O(n^{\frac 3 2})$ for $|M|\geq \frac{3}{8}$. It 
is not clear, whether we can apply the technique of the randomized  associative algorithm by 
Rajagopalan and Schulman~\cite{RS00} in connection with the quantum walk search schema of 
Magniez {\it et al.\/}~\cite{MNRS07}.

Some quantum query lower bound remain open. 
Are we able to prove a nontrivial lower bound for the group 
problem. Our upper bound for this problem is $\tilde{O}(n^{\frac{11}{14}})$. 
It would also be very interesting to close the gap between the $\Omega(n)$ lower bound
and the $O(n^{7/6})$ upper bound for the quasigroup and the loop problem.

   



\begin{thebibliography}{AAAA}

\bibitem[Amb02]{Amb02} A. Ambainis, \textit{Quantum Lower Bounds by Quantum Arguments}, Journal of Computer and System Sciences 64: pages 750-767, 2002.
\bibitem[Amb03]{Amb03} A. Ambainis, \textit{Quantum walks and their algorithmic applications}, International Journal of Quantum Information 1: pages 507-518, 2003.
\bibitem[Amb04]{Amb04} A. Ambainis, \emph{Quantum walk algorithm for element distinctness}, Proceedings of FOCS'04: pages 22-31, 2004.
\bibitem[Amb05]{Amb05} A. Ambainis, \textit{Quantum Search Algorithms}, Technical Report arXiv:quant-ph/0504012, 2005.
\bibitem[AS06]{AS06} A. Ambainis, R. {\v S}palek, \emph{Quantum Algorithms for Matching and Network Flows},
Proceedings of STACS'06, 2006.

\bibitem[BBCMW01]{BBCMW01} R. Beals, H. Buhrman, R. Cleve, M. Mosca, R. de Wolf, \emph{Quantum lower bounds by polynomials}, Journal of ACM 48: pages 778-797, 2001.
\bibitem[BBHT98]{BBHT98} M. Boyer, G. Brassard, P. H\o{}yer, A. Tapp, \emph{Tight bounds on quantum searching}, Fortschritte Der Physik 46(4-5): pages 493-505, 1998.
\bibitem[BDHHMSW01]{BDHHMSW01} H. Buhrman, C. D\"urr, M Heiligman, P. H\o yer, F. Magniez, M. Santha, R. de Wolf, \textit{Quantum Algorithms for Element Distinctness}, Proceedings of CCC'01: pages 131-137, 2001. 

\bibitem[BHMT00]{BHMT00} G. Brassard, P. H\'oyer, M. Mosca, A. Tapp, \textit{Quantum amplitude amplification and estimation}, In Quantum Computation and Quantum Information: A Millennium Volume, AMS Contemporary
Mathematics Series, 2000.
\bibitem[BS06]{BS06} H. Buhrman, R. \v Spalek, \textit{Quantum Verification of Matrix Products}, Proceedings of SODA'06: pages 880-889, 2006.

\bibitem[Doe07a]{Doe07a} S. D\"orn, \textit{Quantum Complexity Bounds of Independent Set Problems}, Proceedings of SOFSEM'07 (SRF): pages 25-36, 2007.
\bibitem[Doe07b]{Doe07b} S. D\"orn, \textit{Quantum Algorithms for Graph Traversals and Related Problems}, Proceedings of CIE'07, 2007.
\bibitem[DHHM04]{DHHM04} C. D\"urr, M. Heiligman, P. H\o yer, M. Mhalla, \textit{Quantum query complexity of some graph problems}, Proceedings of  ICALP'04: pages 481-493, 2004.
\bibitem[Gro96]{Gro96} L. Grover,  \emph{A fast mechanical algorithm for database search},  Proceedings of STOC'96: pages 212-219, 1996.


\bibitem[Kav03]{Kav03}  T. Kavitha, \textit{Efficient Algorithms for Abelian Group Isomorphism and Related Problems}
 Proceedings of FSTTCS'03: pages: 277-288, 2003.
 \bibitem[KSW03]{KSW03} J. Kempe, N. Shenvi, K.B. Whaley, \textit{Quantum Random-Walk Search Algorithm}, 
Physical Review Letters A, Vol. 67 (5), 2003. 


\bibitem[MN05]{MN05} F. Magniez, A. Nayak, \textit{Quantum complexity of testing group commutativity}, Proceedings of ICALP'05: pages 1312-1324, 2005.
\bibitem[MNRS07]{MNRS07}  F. Magniez, A. Nayak, J. Roland, M. Santha, \textit{Search via Quantum Walk}, 
Proceedings of STOC'07, 2007.
\bibitem[MSS05]{MSS05} F. Magniez, M. Santha, M. Szegedy, \emph{Quantum Algorithms for the Triangle Problem}, Proceedings of  SODA'05: pages 1109-1117, 2005.
\bibitem[NC03]{NC03} M.A. Nielsen, I. L. Chuang,  \textit{Quantum Computation and Quantum Information}, Cambridge University Press, 2003.


\bibitem[RS00]{RS00} S. Rajagopalan, L. J. Schulman, \textit{Verification of identities}, SIAM J. Computing 29(4): pages 1155-1163, 2000. 
\bibitem[Sim94]{Sim94} D.R. Simon, \textit{On the power of quantum computation}, Proceedings of FOCS'94: pages 116-123, 1994.
\bibitem[Sze04]{Sze04} M.~Szegedy, \textit{Quantum speed-up of markov chain based algorithms}, 
Proceedings of FOCS'04: pages 32-41, 2004.

\end{thebibliography}
\end{document}